\documentclass[11pt,fleqn]{article}

\usepackage{amsmath,amssymb,amsthm}
\usepackage[margin=1in]{geometry}
\usepackage{xcolor}
\usepackage{array}
\usepackage{comment}
\usepackage{booktabs}
\usepackage{subcaption}
\usepackage{tikz}
\usetikzlibrary{shapes.geometric,decorations.pathreplacing,arrows.meta}
\usepackage[hypertexnames=false,bookmarksnumbered=true,final]{hyperref}
\usepackage[capitalize,sort]{cleveref}

\def\colorschemesepia{sepia}
\def\colorschemedark{dark}
\def\colorschemelight{light}

\ifx\colorscheme\undefined
\let\colorscheme\colorschemelight
\fi

\ifx\colorscheme\colorschemelight
\colorlet{textColor}{black}
\colorlet{bgColor}{white}
\fi

\ifx\colorscheme\colorschemesepia
\definecolor{textColor}{HTML}{433423}
\definecolor{bgColor}{HTML}{fbf0da}
\fi

\ifx\colorscheme\colorschemedark
\definecolor{textColor}{HTML}{bdc1c6}
\definecolor{bgColor}{HTML}{202124}
\definecolor{textBlue}{HTML}{8ab4f8}
\definecolor{textRed}{HTML}{f9968b}
\definecolor{textGreen}{HTML}{81e681}
\definecolor{textPurple}{HTML}{c58af9}
\else
\colorlet{textBlue}{blue!50!black}
\colorlet{textRed}{red!50!black}
\colorlet{textGreen}{green!50!black}
\definecolor{textPurple}{HTML}{681da8}
\fi

\ifx\colorscheme\colorschemelight\else
\pagecolor{bgColor}
\color{textColor}
\fi

\colorlet{dimColor}{textColor!50!bgColor}

\hypersetup{colorlinks,linkcolor=textRed,citecolor=textRed,urlcolor=textBlue}

\title{Proportional and Pareto-Optimal Allocation\\of Chores with Subsidy%
\thanks{Jugal Garg and Eklavya Sharma were supported by NSF Grant CCF-2334461.}}
\author{
Jugal Garg%
\thanks{Department of Industrial \& Enterprise Engineering, University of Illinois at Urbana-Champaign, USA}
\\ \texttt{\small jugal@illinois.edu}
\and
Eklavya Sharma\footnotemark[2]
\\ \texttt{\small eklavya2@illinois.edu}
\and
Xiaowei Wu\thanks{University of Macau}\\ \texttt{\small xiaoweiwu@um.edu.mo}
}
\date{\empty}

\newcommand*{\defeq}{:=}

\newcommand*{\WLoG}{Without loss of generality}
\newcommand*{\wLoG}{without loss of generality}
\newcommand*{\odd}{odd}

\newcommand*{\floor}[1]{\lfloor #1 \rfloor}
\newcommand*{\ceil}[1]{\lceil #1 \rceil}

\DeclareMathOperator*{\E}{\mathbb{E}}

\DeclareMathOperator*{\argmax}{argmax}
\newcommand*{\optprog}[3]{
\begin{array}{>{\displaystyle}c*2{>{\displaystyle}l}}
#1 & \multicolumn{2}{>{\displaystyle}l}{#2}
#3 \end{array}}

\makeatletter
\g@addto@macro{\UrlBreaks}{%
\do\/%
\do\a\do\b\do\c\do\d\do\e\do\f\do\g\do\h\do\i\do\j\do\k\do\l\do\m%
\do\n\do\o\do\p\do\q\do\r\do\s\do\t\do\u\do\v\do\w\do\x\do\y\do\z%
\do\A\do\B\do\C\do\D\do\E\do\F\do\G\do\H\do\I\do\J\do\K\do\L\do\M%
\do\N\do\O\do\P\do\Q\do\R\do\S\do\T\do\U\do\V\do\W\do\X\do\Y\do\Z%
\do\0\do\1\do\2\do\3\do\4\do\5\do\6\do\7\do\8\do\9%
}
\makeatother

\newcommand*{\Ical}{\mathcal{I}}
\newcommand*{\Icalhat}{\widehat{\Ical}}

\newcommand*{\chat}{\widehat{c}}
\newcommand*{\dhat}{\widehat{d}}

\newcommand*{\xhat}{\widehat{x}}

\DeclareMathOperator{\optSub}{optSub}

\DeclareMathOperator{\rcost}{rcost}

\allowdisplaybreaks

\newtheorem{theorem}{Theorem}
\newtheorem{definition}{Definition}
\newtheorem{example}{Example}

\newtheorem{lemma}{Lemma}

\tikzset{
vertex/.style = {draw, inner sep=0, outer sep=0, minimum size=1.6em, fill={textColor!5!bgColor}, thick},
myArrow/.style = {->,>={Stealth}},
agentNode/.style = {vertex,circle},
itemNode/.style = {vertex,rectangle},
rNode/.style = {fill={red!5!bgColor},draw={textRed},text={textRed}},
gNode/.style = {fill={green!5!bgColor},draw={textGreen},text={textGreen}},
bNode/.style = {fill={blue!5!bgColor},draw={textBlue},text={textBlue}},
}

\let\citet\cite
\let\citep\cite

\newcommand*{\ifTwoColumn}[2]{#2}

\begin{document}

\maketitle
\setlength{\parskip}{0.3em}

\begin{abstract}
We consider the problem of allocating $m$ indivisible chores among $n$ agents with possibly different weights, aiming for a solution that is both fair and efficient. Specifically, we focus on the classic fairness notion of proportionality and efficiency notion of Pareto-optimality. Since proportional allocations may not always exist in this setting, we allow the use of \emph{subsidies} (monetary compensation to agents) to ensure agents are proportionally-satisfied, and aim to minimize the total subsidy required.
\texorpdfstring{\citet{wu2024tree}}{Wu and Zhou (WINE 2024)} showed that when each chore has disutility at most 1, a total subsidy of at most $n/3 - 1/6$ is sufficient to guarantee proportionality. However, their approach is based on a complex
technique, which does not guarantee economic efficiency --- a key desideratum in fair division.

In this work, we give a polynomial-time algorithm that achieves the same subsidy bound while also ensuring Pareto-optimality. Moreover, both our algorithm and its analysis are significantly simpler than those of
\texorpdfstring{\citet{wu2024tree}}{Wu and Zhou (WINE 2024)}.
Our approach first computes a proportionally-fair competitive equilibrium, and then applies a rounding procedure guided by minimum-pain-per-buck edges.

\end{abstract}

\section{Introduction}
\label{sec:intro}

We study the problem of allocating $m$ indivisible chores among $n$ agents in a manner that satisfies both fairness and efficiency. It arises in various practical settings, such as assigning tasks among employees or teams, or distributing responsibilities like waste-management and environment-protection among governments or industries. In many such contexts, agents may have different \emph{weights} or \emph{obligations}, reflecting differences in capacity or responsibility. For example, a team with more employees might be expected to handle a larger share of tasks, and an industry that generates more waste should bear a greater burden of waste-management. Formally, each agent $i$ is associated with a weight $w_i$, where the weights sum to 1.

Two classical notions of fairness in such settings are \emph{proportionality} (PROP) and \emph{envy-freeness} (EF).
PROP requires that each agent $i$ receives no more than a $w_i$ fraction of the total burden, according to her own disutility function.
Formally, let $M$ denote the set of chores, and let each agent $i$ have a \emph{disutility function} $d_i: 2^M \to \mathbb{R}_{\ge 0}$, which assigns a cost to each subset of chores.
An allocation $A = (A_1, \ldots, A_n)$ that assigns each agent $i$ the bundle $A_i \subseteq M$ is said to be PROP if $d_i(A_i) \le w_id_i(M)$ for each agent $i$.
Agent $i$ \emph{envies} agent $j$ in an allocation $A$ if $i$ thinks that
$j$'s bundle's disutility is less than $w_j/w_i$ times her own bundle's disutility.
Formally, $i$ envies $j$ if
\[ \frac{d_i(A_i)}{w_i} > \frac{d_i(A_j)}{w_j}. \]
Allocation $A$ is said to be EF if no agent envies another agent.

Although PROP and EF are conceptually compelling fairness notions,
allocations satisfying them do not always exist.
For example, when there are 5 identical chores and two agents with equal weights,
some agent would receive at least 3 chores, and that agent
would not be proportionally-satisifed, and would envy the other agent.
However, we can still aim for \emph{approximate fairness};
giving three chores to one agent and two to the other is, intuitively, as fair as possible.
Hence, notions of approximate fairness have been extensively studied.
These notions are obtained by \emph{weakening} or \emph{relaxing} the definitions of EF and PROP.
Well-known relaxations of EF include EF1 and EFX
\citep{lipton2004approximately,caragiannis2019unreasonable,springer2024almost,chakraborty2021weighted,bhaskar2021on},
and well-known relaxations of PROP include
PROPx \citep{li2022almost}, PROP1 \citep{aziz2020polynomial},
MMS \citep{budish2011combinatorial,huang2023reduction,akrami2024breaking},
and APS \citep{babaioff2023fair}.
\citet{amanatidis2023fair} give a survey of research on approximate fairness notions.

Instead of approximate fairness, one can try to achieve exact fairness through the use of money.
\citet{halpern2019fair} study a setting where the items to be allocated
are goods (i.e., items of non-negative utility) instead of chores,
and each agent must be paid a \emph{subsidy} such that
the resulting allocation of goods and subsidies is envy-free, and the total subsidy is minimized.
Following their work, \emph{envy-free division with subsidy} has been extensively studied
for many different settings
\citep{brustle2020one,barman2022achieving,kawase2024towards,goko2024fair,elmalem2025whoever,goko2024fair}.

\citet{wu2024one,wu2024tree} introduced the closely related \emph{PROP-subsidy} problem, which is the focus of our work.
Here chores and subsidy must be allocated to the agents such that the resulting allocation is PROP.
Formally, if agent $i$ is assigned a bundle $A_i$, she receives a subsidy of $\max(0, d_i(A_i) - w_id_i(M))$,
where $M$ is the set of all chores. Our goal is to find an allocation $A$ that minimizes the total subsidy.
How large can the total subsidy be in the worst case?
The following example shows a lower bound of $n/4$ on the worst-case subsidy.

\begin{example}
\label{thm:prop-subsidy-chores-hard}
Consider $n$ agents with equal weights and $n/2$ chores,
where each agent's disutility for a set $S$ of chores is given by $|S|$.
In this instance, each agent's proportional share is $1/2$,
and any allocation results in a total subsidy of at least $n/4$.
\end{example}

When agents have equal weights and additive disutility functions%
\footnote{\label{foot:additive}A function $f: 2^M \to \mathbb{R}$ is called \emph{additive}
if $f(S) = \sum_{c \in S} f(\{c\})$ for all $S \subseteq M$.},
and each chore has disutility of at most 1 for each agent,
\citet{wu2024one} gave a polynomial-time allocation algorithm
that guarantees a total subsidy of at most $n/4$, thus matching the lower bound.
For unequal weights, they gave an algorithm with a total subsidy of at most $(n-1)/2$.
Subsequently, \citet{wu2024tree} improved this result by
designing an intricate algorithm that reduces the total subsidy to at most $n/3 - 1/6$.

In addition to fairness, achieving economic efficiency is a key objective in fair division.
\emph{Pareto optimality} is the standard notion used to capture economic efficiency.
An allocation $A$ is said to \emph{Pareto-dominate} another allocation $B$
if no one prefers $B$ and someone prefers $A$, i.e.,
$d_i(A_i) \le d_i(B_i)$ for every agent $i$ and $d_i(A_i) < d_i(B_i)$ for some agent $i$.
An allocation is \emph{Pareto-optimal} (PO) if it is not Pareto-dominated by any other allocation.
\Cref{ex:prop-not-po} illustrates the importance of efficiency.

\begin{example}
\label{ex:prop-not-po}
Consider two agents with equal weights and additive disutility functions.
There are four chores whose disutilities are given by the following table.

\begin{center}
\begin{tabular}{ccccc}
\toprule & $c_1$ & $c_2$ & $c_3$ & $c_4$
\\ \midrule
$d_1(\cdot)$ & $1$ & $1$ & $100$ & $100$
\\ $d_2(\cdot)$ & $100$ & $100$ & $1$ & $1$
\\ \bottomrule
\end{tabular}
\end{center}

Each agent's proportional share is 101.
The allocation $(\{c_1, c_3\}, \{c_2, c_4\})$ gives each agent a disutility of 101, so it is PROP.
However, it is not PO, since it is Pareto-dominated by another allocation $(\{c_1, c_2\}, \{c_3, c_4\})$,
where each agent gets a disutility of 2.
Clearly, the second allocation is significantly better than the first.
\end{example}

Obtaining PO along with approximate fairness has been an area of intense research
(e.g., \citep{barman2018finding,aziz2020polynomial,garg2022fair,garg2023new,lin2025approximately,mahara2025existence}),
which gained significant momentum following the Kalai-Prize-winning \cite{kalai2024} work by
\citet{caragiannis2019unreasonable} showing the existence of allocations that are both EF1 and PO.
However, to the best of our knowledge, achieving PO with low subsidy
has not been studied before for either EF or PROP.
For EF-subsidy, it is unclear whether a PO allocation with $O(n)$ subsidy exists.
For PROP-subsidy, existence is clear: if allocation $A$ Pareto-dominates allocation $B$,
then the PROP-subsidy of $A$ is at most that of $B$.
However, finding a Pareto-dominator of an allocation is coNP-hard \cite{dekeijzer2009on},
so finding a low-subsidy PO allocation in polynomial time has been an open problem.

\subsection{Our Contribution}

In this work, we study the PROP-subsidy problem in the general setting
of unequal weights and additive disutilities.
We show that the best-known upper bound of $n/3 - 1/6$
on the total subsidy by \citet{wu2024tree} can be obtained via a Pareto-optimal allocation,
and that such an allocation can be computed in polynomial time.
Thus, we achieve efficiency without compromising on fairness.
Moreover, our results are obtained using significantly simpler techniques;
\citet{wu2024tree}'s paper is around 30 pages long,
of which 11 pages are dedicated to proving the correctness of their rounding scheme.

We achieve efficiency and simplicity through a clever use of \emph{market equilibrium}.
Known algorithms for minimizing subsidy \citep{wu2024one,wu2024tree} proceed in two phases.
\begin{enumerate}
\item \textbf{Fractional allocation}:
    Compute a \emph{fractional} proportional allocation $x$, where each agent $i$ is assigned a fraction $x_{i,c}$ of each chore $c$, ensuring proportionality.
\item \textbf{Rounding}:
    \emph{Round} $x$ into an integral allocation $A$, where agent $i$ receives chore $c$ only if $x_{i,c} > 0$.
    The goal is to keep the rounding cost low, i.e., the disutility of each agent's bundle in $A$ should be, on average, close to the disutility in the fractional allocation $x$.
\end{enumerate}

Prior works \citep{wu2024one,wu2024tree} use a greedy algorithm called \emph{fractional-bid-and-take} in the first phase. In contrast, we compute a fractional allocation by finding a competitive (market) equilibrium. A market equilibrium is a pair $(x, p)$ where $x$ is a fractional allocation and $p \in \mathbb{R}_{>0}^m$ is a \emph{payment vector}, where $p$ assigns to each chore a \emph{payment} representing its aggregate disutility across all agents. Market equilibria are well-known to be closely related to Pareto-optimality
\citep{mas1995microeconomic16,branzei2024algorithms,bogomolnaia2017competitive},
and have been used extensively to obtain fairness and efficiency simultaneously in allocation problems
\citep{barman2018finding,aziz2020polynomial,garg2022fair,garg2023new,lin2025approximately,mahara2025existence}.

The key to our algorithm's success lies in the introduction of a novel intermediate step between the fractional allocation and rounding, which we call \emph{reduction}. This step reduces the rounding problem to a special case where all agents have identical disutility functions, significantly simplifying the rounding procedure. The reduction step crucially utilizes the payment vector from the market equilibrium, demonstrating the utility of market equilibria beyond their traditional connection to Pareto-optimality.

\subsection{Related Work}

Unless specified otherwise, we assume that all agents have additive disutilities.

\paragraph{PROP-Subsidy for Goods.}
An analogous problem is the fair allocation of goods,
where the items have \emph{utility} or \emph{value}.
Formally, $M$ is the set of all goods,
and each agent $i$ has a \emph{valuation function} $v_i: 2^M \to \mathbb{R}_{\ge 0}$.
Allocation $A$ is proportional if $v_i(A_i) \ge w_iv_i(M)$.
Agent $i$'s subsidy in $A$ is $\max(0, w_iv_i(M) - v_i(A_i))$.
Prior work on PROP-subsidy also holds for goods;
when each good has a value of at most 1 to each agent,
$n/4$ subsidy suffices for equal weights \citep{wu2024one}
and $n/3 - 1/6$ subsidy suffices for unequal weights \citep{wu2024tree}.
One can show that the lower bound of $n/4$ on the worst-case total subsidy
also holds for goods by appropriately modifying \cref{thm:prop-subsidy-chores-hard}.

\paragraph{EF-Subsidy.}
In the EF-subsidy problem, we pay each agent $i$ a subsidy of $s_i$
to make envy disappear, and we want to minimize the total subsidy $\sum_{i=1}^n s_i$.
Formally, the subsidy vector $(s_1, \ldots, s_n)$ is feasible for allocation $A$ if
for any two agents $i$ and $j$, we have
\[ \frac{d_i(A_i) - s_i}{w_i} \le \frac{d_i(A_j) - s_j}{w_j}. \]
We assume that each chore has a disutility of at most 1 to each agent.

Note that EF-subsidy and PROP-subsidy have a fundamental difference:
in the EF-subsidy problem, agents care about the subsidies paid to other agents,
whereas in the PROP-subsidy problem, they don't.

If there are $n-1$ identical chores and $n$ agents, a subsidy of $n-1$ is required to get EF.
For chores and equal agent weights, \citet{wu2024one} give an algorithm that
achieves EF with a subsidy of at most $n-1$, thus matching the lower bound.

One can define envy-freeness and EF-subsidy for goods analogously.
In fact, EF-subsidy is much more well-studied for goods than for chores.
If there is just one good, a subsidy of $n-1$ is required.
When agents have equal weights, \citet{halpern2019fair} show that
a subsidy of $m(n-1)$ suffices to get EF.
\citet{brustle2020one} improved the upper-bound to $n-1$, thus matching the lower bound.
\citet{barman2022achieving,kawase2024towards,goko2024fair} study EF-subsidy
when agents' valuation functions may not be additive.
\citet{elmalem2025whoever} study EF-subsidy and unequal agent weights.
\citet{goko2024fair} study EF-subsidy with truthfulness.

\paragraph{Relaxations of EF and PROP}
When money cannot be used to achieve fairness, notions of approximate fairness are used.
These notions are obtained by \emph{weakening} or \emph{relaxing} the definitions of EF and PROP.
\citet{amanatidis2023fair} give a survey of research on approximate fairness notions.

EF was initially relaxed to a notion called \emph{EF1} (envy-freeness up to one item),
and algorithms were designed to guarantee EF1 allocations
\citep{budish2011combinatorial,lipton2004approximately,bhaskar2021on,chakraborty2021weighted,springer2024almost}.
A stronger relaxation of EF, called \emph{EFX} (envy-free up to any item),
was introduced in \citet{caragiannis2019unreasonable}.
For unequal weights, EFX is infeasible \citep{springer2024almost}, and for equal weights,
despite significant efforts, the existence of EFX allocations remains an open problem.
Consequently, relaxations of EFX \citep{caragiannis2023new,amanatidis2020multiple,chaudhury2021little}
as well as the existence of EFX in special cases \citep{chaudhury2024efx,plaut2020almost,amanatidis2021maximum}
have been explored.

PROP1 (proportional up to one item) is a well-known relaxation of PROP,
and its existence was easy to prove \citep{aziz2021fair}.
For chores, a stronger notion called PROPx (proportional up to any item)
was also shown to be feasible \citep{li2022almost},
but for goods, PROPx was shown to be infeasible \cite{aziz2020polynomial}, even for equal weights.
Hence, notions like PROPm \citep{baklanov2021achieving,baklanov2021propm}
and PROPavg \citep{kobayashi2025proportional}, which are stronger than PROP1
but weaker than PROPx were shown to be feasible.
For equal weights, MMS (maximin share) \citep{budish2011combinatorial} is another relaxation of PROP
that was (surprisingly) shown to be infeasible \citep{kurokawa2018fair,feige2022tight},
so its relaxations have been studied
\citep{kurokawa2018fair,ghodsi2018fair,barman2020approximation,akrami2024breaking,akrami2023randomized,caragiannis2023new}.
AnyPrice Share (APS) \citep{babaioff2023fair} is another relaxation of PROP
that is inspired by MMS but was designed specifically to handle unequal agent weights.

Most fairness notions were originally defined for the case of equal agent weights and goods.
\citet{garg2025exploring} extend definitions of such fairness notions to more general settings
and explore how different notions are related to each other.

\paragraph{Mixed Divisible and Indivisible Items.}
Subsidy can be seen as a divisible good.
Many works have studied the allocation of a mix of indivisible and divisible goods,
by appropriately relaxing EF and PROP for this setting.
\citet{liu2024mixed} give a survey of such results.

\section{Preliminaries}
\label{sec:prelims}

For any $t \in \mathbb{Z}_{\ge 0}$, we use $[t]$ to denote the set $\{1, 2, \ldots, t\}$.
For any $x \in \mathbb{R}$, we denote $\max(0, x)$ by $(x)^+$.

\subsection{Fair Division Instances and Allocations}

Formally, our input is a \emph{fair division instance}
$\Ical$, which is denoted by $(N, M, (d_i)_{i \in N}, (w_i)_{i \in N})$.
Here $N$ denotes the set of \emph{agents}, $M$ denotes the set of \emph{chores},
$d_i: 2^M \to \mathbb{R}_{\ge 0}$ is called agent $i$'s \emph{disutility function}
(i.e., $d_i(S)$ is a measure of how much agent $i$ dislikes the set $S$ of chores),
and $w_i \in \mathbb{R}_{>0}$ is called agent $i$'s \emph{weight} or \emph{obligation}.
We have $\sum_{i \in N} w_i = 1$.

We often assume \wLoG{} that $N = [n]$ and $M = [m]$.
A function $f: 2^{[m]} \to \mathbb{R}$ is said to be \emph{additive} if
$f(S) = \sum_{c \in S} f(\{c\})$ for all $S \subseteq [m]$.
In this work, we assume that all disutility functions are additive.
For notational convenience, for $c \in [m]$, we often write $d_i(c)$ instead of $d_i(\{c\})$.
Moreover, we assume that $d_i(c) > 0$ for each agent $i$ and chore $c$
(otherwise we can just allocate $c$ to an agent $i$ where $d_i(c) = 0$
and remove $c$ from further consideration.)

A fair division instance is \emph{bounded} if $d_i(c) \le 1$
for every agent $i$ and every chore $c$.
We can assume \wLoG{} that the fair division instance given to us as input is bounded,
since we can scale the disutilities by any factor.

For instance $\Ical \defeq ([n], [m], (d_i)_{i=1}^n, (w_i)_{i=1}^n)$,
a \emph{fractional allocation} $x \in [0, 1]^{n \times m}$ is a matrix where
$\sum_{i=1}^n x_{i,c} = 1$ for each $c \in [m]$.
$x_{i,c}$ denotes the fraction of chore $c$ allocated to agent $i$.
Define $x_i \defeq (x_{i,c})_{c=1}^m$ to be agent $i$'s \emph{bundle} in $x$.
For any agent $i \in [n]$ and vector $z \in \mathbb{R}^m$,
define $d_i(z) \defeq \sum_{c=1}^m (z_c\cdot d_i(c))$.

The \emph{consumption graph} of a fractional allocation $x \in [0, 1]^{n \times m}$, denoted as $G(x)$,
is an undirected bipartite graph $([n], [m], E)$, where $(i, c) \in E \iff x_{i,c} > 0$.

A fractional allocation $x$ is called \emph{integral}
if $x_{i,c} \in \{0, 1\}$ for all $i \in [n]$ and $c \in [m]$.
Alternatively, we can express the integral allocation as a tuple $A \defeq (A_1, \ldots, A_n)$,
where $A_i \defeq \{c \in [m]: x_{i,c} = 1\}$ is the set of chores allocated to agent $i$,
called agent $i$'s \emph{bundle}.

An integral allocation $y \in \{0, 1\}^{n \times m}$ is said to be a \emph{rounding}
of a fractional allocation $x \in [0, 1]^{n \times m}$ iff for all $i \in [n]$ and $c \in [m]$,
we have $y_{i,c} = 1 \implies x_{i,c} > 0$.
Equivalently, $y$ is a rounding of $x$ iff $G(y)$ is a subgraph of $G(x)$.

\subsection{Fairness and Subsidy}

Let $x$ be a fractional allocation for the fair division instance
$\Ical \defeq ([n], [m], (d_i)_{i=1}^n, (w_i)_{i=1}^n)$.
Then $x$ is said to be \emph{proportional} or \emph{PROP-fair} if
$d_i(x_i) \le w_id_i([m])$ for all $i \in [n]$.

In a fractional allocation $x$, the PROP-subsidy of agent $i$ is defined to be $(d_i(x_i) - w_id_i([m]))^+$.
This is the amount of money that must be paid to agent $i$
to compensate for her bundle not being proportionally-fair to her.
The PROP-subsidy of $x$ is the sum of all agents' subsidies, i.e.,
\[ \sum_{i=1}^n (d_i(x_i) - w_id_i([m]))^+. \]
The optimal PROP-subsidy of $\Ical$, denoted as $\optSub(\Ical)$,
is defined as the minimum PROP-subsidy across all integral allocations.
Our goal is to find upper and lower bounds on $\optSub(\Ical)$ assuming
$d_i(c) \le 1$ for all $i \in [n]$ and $c \in [m]$.
Moreover, we would like to design a polynomial-time algorithm whose output
is an allocation with low PROP-subsidy.

\paragraph{Lower bound on subsidy:}
Consider a fair division instance $\Ical$ with $n$ equally-entitled agents and $m < n$ identical chores,
where $d_i(c) = 1$ for each chore $c$ and agent $i$. Then $\optSub(\Ical) = m(n-m)/n$.
This is maximized when $m \in \{\floor{n/2}, \ceil{n/2}\}$:
for even $n$, we get $\optSub(\Ical) = n/4$,
and for odd $n$, we get $\optSub(\Ical) = (n/4)(1-n^{-2})$.

One way to obtain low PROP-subsidy is to round a fractional PROP allocation.
Let $x \in [0, 1]^{n \times m}$ be a fractional allocation for instance
$\Ical \defeq ([n], [m], (d_i)_{i=1}^n, (w_i)_{i=1}^n)$.
Let integral allocation $A$ be a rounding of $x$.
The \emph{cost of rounding} $x$ to $A$ is defined as
\[ \rcost_{\Ical}(A \mid x) \defeq \sum_{i=1}^n (d_i(A_i) - d_i(x_i))^+. \]

\subsection{Pareto-Optimality and Market Equilibria}
\label{sec:prelims:po-meq}

\begin{definition}[Pareto-optimality]
For a fair division instance $\Ical \defeq ([n], [m], (d_i)_{i=1}^n, (w_i)_{i=1}^n)$,
a fractional allocation $x$ is said to \emph{Pareto-dominate} another fractional allocation $y$ if
$d_i(x_i) \le d_i(y_i)$ for all $i \in [n]$ and $d_i(x_i) < d_i(y_i)$ for some $i \in [n]$.
An allocation is fractionally-Pareto-optimal (fPO) if it is not Pareto-dominated by any fractional allocation.
An integral allocation is Pareto-optimal (PO) if it is not Pareto-dominated by any integral allocation.
(Note that every integral fPO allocation is also PO.)
\end{definition}

Working with PO allocations is difficult, since very few mathematical characterizations of them are known,
and it is often computationally hard to find such allocations.
However, (fractional) fPO allocations can be found in polynomial time,
and there is a rich literature connecting them to \emph{market equilibria}.

\begin{definition}[Market equilibrium]
For a fair division instance $\Ical \defeq ([n], [m], (d_i)_{i=1}^n, (w_i)_{i=1}^n)$,
let $x$ be a fractional allocation and $p \in \mathbb{R}_{>0}^m$ ($p$ is called the \emph{payment vector}).
Then $(x, p)$ is called a \emph{market equilibrium} if
for every agent $i \in [n]$, there exists $\alpha_i \in \mathbb{R}_{>0}$ such that
$d_i(c) \ge \alpha_i p_c$ for all $c \in [m]$,
and $d_i(c) = \alpha_i p_c$ when $x_{i,c} > 0$.
($\alpha_i$ is called $i$'s \emph{minimum-pain-per-buck}.)
\end{definition}

For goods, fPO allocations are related to market equilibria
via the \emph{Welfare theorems} \citep{mas1995microeconomic16}.
For chores, the first welfare theorem is known to hold.

\begin{theorem}[First Welfare Theorem \cite{bogomolnaia2017competitive}]
\label{thm:first-welfare}
Let $(x, p)$ be a market equilibrium for the instance
$\Ical \defeq ([n], [m], (d_i)_{i=1}^n, (w_i)_{i=1}^n)$.
Then $x$ is an fPO allocation.
\end{theorem}
\begin{proof}
Suppose $x$ is not fPO. Then a fractional allocation $y$ Pareto-dominates $x$.
Hence, for all $i \in [n]$, we get
\[ p(y_i) \le \frac{d_i(y_i)}{\alpha_i} \le \frac{d_i(x_i)}{\alpha_i} = p(x_i). \]
Moreover, $d_i(y_i) < d_i(x_i)$ for some agent $i \in [n]$. Then
\[ p(y_i) \le \frac{d_i(y_i)}{\alpha_i} < \frac{d_i(x_i)}{\alpha_i} = p(x_i). \]
Hence, $p([m]) = \sum_{i=1}^n p(y_i) < \sum_{i=1}^n p(x_i) = p([m])$,
which is a contradiction. Hence, $x$ is fPO.
\end{proof}

\section{Technical Overview}
\label{sec:overview}

Our alorithm for obtaining an fPO allocation with low subsidy
starts by computing a PROP+fPO fractional allocation \cref{thm:prop-meq}.

\begin{theorem}
\label{thm:prop-meq}
For any fair division instance over chores, we can compute a
fractional allocation $x$ and payments $p$ in polynomial time such that
$(x, p)$ is a market equilibrium, $x$ is proportional, and $G(x)$ is acyclic.
\end{theorem}
\begin{proof}[Proof sketch]
Let $x$ be an optimal extreme-point solution to the following LP:
\begin{equation}
\label{eqn:lp-primal}
\optprog{\min_{x \in \mathbb{R}_{\ge 0}^{n \times m}}}{\sum_{i=1}^n\sum_{c=1}^m d_i(c)x_{i,c}}{%
\\[1.3em] \text{where} & d_i(x_i) \le w_id_i([m]) & \forall i \in [n]
\\ \text{and} & \sum_{i=1}^n x_{i,c} = 1 & \forall c \in [m].
}
\end{equation}
Then one can show that $x$ is proportional and $G(x)$ is acyclic.
One can obtain the corresponding payments $p$ by solving the LP's dual.
See \cref{sec:prop-meq-over-forest} for a formal proof.
\end{proof}

Next, we would like to round the allocation obtained from \cref{thm:prop-meq}.
The following result shows how to reduce the problem of rounding this fractional allocation
to the special case of identical disutilities.

\begin{theorem}
\label{thm:reduce-to-idval}
Let $\Ical \defeq ([n], [m], (d_i)_{i=1}^n, (w_i)_{i=1}^n)$ be a bounded fair division instance.
Let $(x, p)$ be a market equilibrium for $\Ical$.
For every chore $c \in [m]$, define $\dhat(c) \defeq p_c/p_{\max}$,
where $p_{\max} \defeq \max_{j=1}^m p_j$.
Let $\Icalhat$ be the fair division instance $([n], [m], (\dhat)_{i=1}^n, (w_i)_{i=1}^n)$
(i.e., every agent has disutility function $\dhat$).
Then $\Icalhat$ is bounded, and for every rounding $A$ of $x$, we have
$\rcost_{\Icalhat}(A \mid x) \ge \rcost_{\Ical}(A \mid x)$.
\end{theorem}

Before we prove \cref{thm:reduce-to-idval}, let us look at an example

\begin{example}
\label{ex:reduce-to-idval}
Consider a fair division instance $\Ical$ with three agents $a_1$, $a_2$, and $a_3$,
having weights $2/15$, $8/15$, and $1/3$, respectively.
There are three chores $c_1$, $c_2$, and $c_3$, having the following disutilities:

\begin{center}
\begin{tabular}{cccc}
\toprule & $c_1$ & $c_2$ & $c_3$
\\ \midrule
$d_1(\cdot)$ & $1/2$ & $1$ & $1/2$
\\ $d_2(\cdot)$ & $1$ & $1$ & $1$
\\ $d_3(\cdot)$ & $1$ & $2/3$ & $1/3$
\\ \bottomrule
\end{tabular}
\end{center}

Then for $p = (1, 1, 1/2)$ and
\[ x = \begin{pmatrix}
1/3 & \textcolor{dimColor}{0} & \textcolor{dimColor}{0}
\\ 2/3 & 2/3 & \textcolor{dimColor}{0}
\\ \textcolor{dimColor}{0} & 1/3 & 1
\end{pmatrix}, \]
$(x, p)$ is a market equilibrium, with $\alpha = (1/2, 1, 2/3)$
as the minimum-pain-per-buck vector.
Also, $x$ is PROP, since the agents' values for their own bundles are
$(1/6, 4/3, 5/9)$, and their proportional shares are $(4/15, 8/5, 2/3)$.
$G(x)$ is acyclic, and looks like this
(edge between agent $i$ and chore $c$ is labeled with $x_{i,c}$):

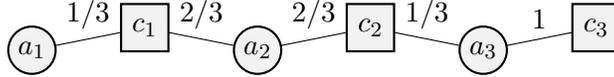
\begin{center}
\begin{tikzpicture}
\node[agentNode] (a1) at (0.0, 0.0) {$a_1$};
\node[itemNode]  (c1) at (1.5, 0.3) {$c_1$};
\node[agentNode] (a2) at (3.0, 0.0) {$a_2$};
\node[itemNode]  (c2) at (4.5, 0.3) {$c_2$};
\node[agentNode] (a3) at (6.0, 0.0) {$a_3$};
\node[itemNode]  (c3) at (7.5, 0.3) {$c_3$};

\draw (a1)
    -- node[midway, above] {$1/3$} (c1)
    -- node[midway, above] {$2/3$} (a2)
    -- node[midway, above] {$2/3$} (c2)
    -- node[midway, above] {$1/3$} (a3)
    -- node[midway, above] {$1$} (c3);
\end{tikzpicture}

\end{center}

By \cref{thm:reduce-to-idval}, we get a new instance $\Icalhat$
where the chores have disutilities $(1, 1, 1/2)$ for every agent.
There are 4 possible ways to round $x$ to an integral allocation,
and by inspecting all of them, we conclude that
the allocation $A = (\{\}, \{c_1\}, \{c_2, c_3\})$
incurs a rounding cost of $2/3$ in $\Icalhat$, which is optimal.
Note that $2/3$ is less than the worst-case subsidy bound of $n/3 - 1/6$.

\Cref{thm:reduce-to-idval} tells us that $\rcost_{\Ical}(A \mid x) \le \rcost_{\Icalhat}(A \mid x)$.
Indeed, $\rcost_{\Ical}(A \mid x) = 4/9 < 2/3$.
Moreover, the PROP-subsidy of $A$ is $1/3$, which is less than $4/9$.
\end{example}

\begin{proof}[Proof of \cref{thm:reduce-to-idval}]
Let $\alpha_i$ be the minimum-pain-per-buck for each agent $i$.
For each pair $(i, c)$ such that $x_{i,c} = 0$,
we first reduce $d_i(c)$ such that $d_i(c) = \alpha_i p_c$.
Then each agent $i$'s disutility function becomes proportional to the payment vector $p$,
but $d_i(A_i)$ and $d_i(x_i)$ remain unchanged,
so $\rcost(A \mid x)$ remains unchanged.
Now for each agent $i$, increase $d_i(c)$ to $\max_{i \in [n]} d_i(c)$
so that disutility functions become identical.
Then for each agent $i \in [n]$,
$(d_i(A_i) - d_i(x_i))^+$ either stays the same or increases.

The above process can be succinctly emulated by setting the disutility
of each chore $c$ to $p_c/p_{\max}$. We now formally prove its correctness.

Let $c^* \in \argmax_{c \in [m]} p_c$.
Then for any $i \in [n]$, we have
\begin{equation*}
    \alpha_i p_{c^*} \le d_i(c^*) \le 1.
\end{equation*}
For any $i \in [n]$ and $c \in [m]$ such that $x_{i,c} > 0$,
we have $d_i(c) = \alpha_i p_c$. Therefore,
\[ \frac{d_i(c)}{\dhat(c)} = \frac{\alpha_ip_c}{p_c/p_{c^*}} = \alpha_ip_{c^*} \le 1. \]
Hence,
\begin{align*}
& (\dhat(A_i) - \dhat(x_i))^+
    = \frac{(d_i(A_i) - d_i(x_i))^+}{\alpha_i p_{c^*}}
\ifTwoColumn{\\&\quad}{}
    \ge (d_i(A_i) - d_i(x_i))^+.
\end{align*}
Hence, $\rcost_{\Icalhat}(A \mid x) \ge \rcost_{\Ical}(A \mid x)$.
\end{proof}

Next, in \cref{sec:low-cost-rounding}, we prove the following theorem,
which solves the low-cost rounding problem.

\begin{theorem}
\label{thm:lcr}
Let $\Icalhat$ be a bounded fair division instance
where every agent has the same disutility function.
Let $x$ be any fractional allocation for $\Icalhat$ such that $G(x)$ is acyclic.
Then we can compute a rounding $A$ of $x$ in $O(m+n^2)$ time such that
$\rcost_{\Icalhat}(A \mid x) \le n/3 - 1/6$.
\end{theorem}

Finally, we combine these results to prove our main result.

\begin{theorem}
\label{thm:low-subsidy}
Let $\Ical \defeq ([n], [m], (d_i)_{i=1}^n, (w_i)_{i=1}^n)$ be a bounded fair division instance.
Then we can, in polynomial time, compute an integral fPO allocation $A$
having PROP-subsidy at most $n/3 - 1/6$, i.e.,
\[ \sum_{i=1}^n (d_i(A_i) - w_id_i([m]))^+ \le \frac{n}{3} - \frac{1}{6}. \]
\end{theorem}
\begin{proof}
Let $(x, p)$ and $\Icalhat$ be as defined by \cref{thm:prop-meq,thm:reduce-to-idval}, respectively.
We can compute them in polynomial time.
Using \cref{thm:lcr}, we can compute a rounding $A$ of $x$ in polynomial time such that
$\rcost_{\Icalhat}(A \mid x) \le n/3 - 1/6$. Hence, the total subsidy is at most
\begin{align*}
& \sum_{i=1}^n (d_i(A_i) - w_id_i([m]))^+
\\ &\le \sum_{i=1}^n (d_i(A_i) - d_i(x_i))^+
    \tag{since $x$ is proportional}
\\ &= \rcost_{\Ical}(A \mid x)
    \le \rcost_{\Icalhat}(A \mid x)
    \tag{by \cref{thm:reduce-to-idval}} \\
    & \le \frac{n}{3} - \frac{1}{6}.
    \tag{by \cref{thm:lcr}}
\end{align*}
Moreover, since $A$ is a rounding of $x$, $(A, p)$ is also a market equilibrium.
By the first welfare theorem (\cref{thm:first-welfare} in \cref{sec:prelims:po-meq}),
we get that $A$ is fPO.
\end{proof}

\section{An Acyclic Proportional Market Equilibrium}
\label{sec:prop-meq-over-forest}

Let $\Ical \defeq ([n], [m], (d_i)_{i=1}^n, (w_i)_{i=1}^n)$ be a fair division instance.
We want to find a market equilibrium $(x^*, p^*)$ such that $x^*$ is PROP and $G(x^*)$ is acyclic.

Any optimal solution to linear program \eqref{eqn:lp-primal} (c.f.~\cref{sec:overview}) is PROP and fPO.
We will show that there exists an optimal solution $x^*$ to the LP such that $G(x^*)$ is acyclic,
and we will show how to find payments $p^*$ such that $(x^*, p^*)$ is a market equilibrium.

\begin{lemma}
\label{thm:lp-in-polytime}
LP \eqref{eqn:lp-primal} is feasible. Moreover, we can obtain in polynomial time
a solution $x^*$ to LP \eqref{eqn:lp-primal} such that $G(x^*)$ is acyclic.
\end{lemma}
\begin{proof}
Let $x^{(0)} \in [0, 1]^{n \times m}$, where $x^{(0)}_{i,c} = w_i$ for all $i \in [n]$ and $c \in [m]$.
Then $x^{(0)}$ is a feasible solution to LP \eqref{eqn:lp-primal}.

We can compute an optimal solution $\xhat$ to the LP in polynomial time.
Lemma 2.5 by \citet{sandomirskiy2022efficient} shows how to transform $\xhat$
to another optimal solution $x^*$ in $O(n^2m^2(m+n))$ time such that $G(x^*)$ is acyclic.
Alternatively, one can show that if $\xhat$ is an extreme-point optimal solution to the LP,
then $G(\xhat)$ is acyclic. The proof is very similar to Lemma 2.5 by \citet{sandomirskiy2022efficient}.
\end{proof}

\citet{branzei2024algorithms} prove the \emph{Second welfare theorem} for chores,
which says that for any fPO allocation $x$, there exist payments $p$ such that $(x, p)$ is a market equilibrium.
However, their result only works for cases where $d_i(x_i) > 0$ for all $i \in [n]$,
which is not guaranteed for $x^*$ obtained as in \cref{thm:lp-in-polytime}.
Hence, instead of relying on their result, we directly (and constructively)
prove the existence of payments $p^*$ such that $(x^*, p^*)$ is a market equilibrium.
We do this by considering the dual of LP \eqref{eqn:lp-primal}:
\begin{equation}
\label{eqn:lp-dual}
\optprog{\max_{p \in \mathbb{R}^m,\,h \in \mathbb{R}_{\ge 0}^n}}{\sum_{c=1}^m p_c - \sum_{i=}^n h_iw_id_i([m])}{%
\\ \text{where} & p_c \le (1+h_i)d_i(c) & \forall i \in [n], \forall c \in [m]
}
\end{equation}

\begin{lemma}
Let $x^*$ be an optimal solution to LP \eqref{eqn:lp-primal}
and $(p^*, h^*)$ be an optimal solution to LP \eqref{eqn:lp-dual}.
Then $(x^*, p^*)$ is a market equilibrium.
\end{lemma}
\begin{proof}
For all $i \in [n]$, let $\alpha_i \defeq 1/(1 + h^*_i)$.
For all $i \in [n]$ and $c \in [m]$, feasibility of $(p^*, h^*)$ gives us $d_i(c) \ge \alpha_ip^*_c$.
When $x^*_{i,c} > 0$, we get $p^*_c - (1+h^*_i)d_i(c) = 0$ by complementary slackness, so $d_i(c) = \alpha_ip^*_c$.
For every $c \in [m]$, there exists $i \in [n]$ such that $x^*_{i,c} > 0$, so $p_c = d_i(c)/\alpha_i > 0$.
Hence, $(x^*, p^*)$ is a market equilibrium.
\end{proof}

\section{Low-Cost Rounding}
\label{sec:low-cost-rounding}

In this section, we prove \cref{thm:lcr}.
Given a fractional allocation $x$ such that $G(x)$ is acyclic, we partition $G(x)$ into
a collection of small chore-disjoint trees and round each tree independently.
See \cref{fig:decomp} for a visual representation.
This technique is similar to the \emph{tree-splitting} technique by \citet{wu2024tree},
but differs in two important ways.
(i) We apply tree-splitting to $G(x)$, whereas \citet{wu2024tree} apply it to
    a forest over agents that is defined in terms of the \emph{fractional bid-and-take algorithm}.
(ii) We partition $G(x)$ into trees differently.
    By exploiting the fact that all agents have the same disutility function,
    the trees in our partition are smaller and have a simpler structure.
    This considerably simplifies our rounding algorithm.

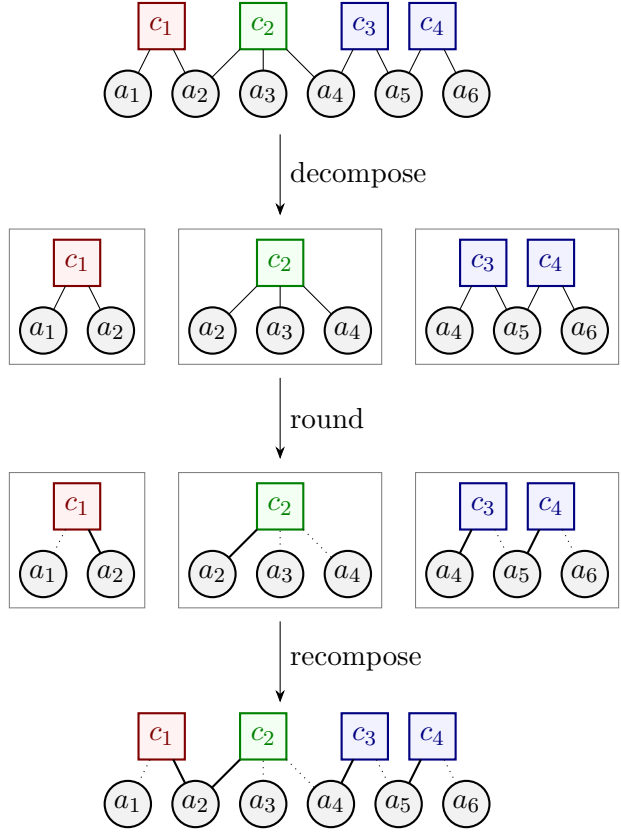
\begin{figure}[htb]
\centering
\begin{tikzpicture}[scale=0.9,every pic/.style={scale=0.9}]

\tikzset{fullGraph/.pic={
\node[agentNode] (a1) at (0.0, 0.0) {$a_1$};
\node[agentNode] (a2) at (1.0, 0.0) {$a_2$};
\node[agentNode] (a3) at (2.0, 0.0) {$a_3$};
\node[agentNode] (a4) at (3.0, 0.0) {$a_4$};
\node[agentNode] (a5) at (4.0, 0.0) {$a_5$};
\node[agentNode] (a6) at (5.0, 0.0) {$a_6$};
\node[itemNode,rNode] (c1) at (0.5, 1.0) {$c_1$};
\node[itemNode,gNode] (c2) at (2.0, 1.0) {$c_2$};
\node[itemNode,bNode] (c3) at (3.5, 1.0) {$c_3$};
\node[itemNode,bNode] (c4) at (4.5, 1.0) {$c_4$};
\draw[lightEdge] (a1) -- (c1);
\draw[heavyEdge] (c1) -- (a2) -- (c2);
\draw[lightEdge] (a3) -- (c2) -- (a4);
\draw[heavyEdge] (a4) -- (c3);
\draw[lightEdge] (a5) -- (c3);
\draw[heavyEdge] (a5) -- (c4);
\draw[lightEdge] (a6) -- (c4);
}}

\tikzset{decompGraph/.pic={
\draw[dimColor] (-0.5, -0.5) rectangle (1.5, 1.5);
\node[agentNode] (a1)  at (0.0, 0.0) {$a_1$};
\node[agentNode] (a2a) at (1.0, 0.0) {$a_2$};
\node[itemNode,rNode] (c1) at (0.5, 1.0) {$c_1$};
\draw[lightEdge] (c1) -- (a1);
\draw[heavyEdge] (c1) -- (a2a);

\draw[dimColor] (2.0, -0.5) rectangle (5.0, 1.5);
\node[agentNode] (a2b) at (2.5, 0.0) {$a_2$};
\node[agentNode] (a3)  at (3.5, 0.0) {$a_3$};
\node[agentNode] (a4a) at (4.5, 0.0) {$a_4$};
\node[itemNode,gNode] (c2) at (3.5, 1.0) {$c_2$};
\draw[heavyEdge] (c2) -- (a2b);
\draw[lightEdge] (a3) -- (c2) -- (a4a);

\draw[dimColor] (5.5, -0.5) rectangle (8.5, 1.5);
\node[agentNode] (a4b) at (6.0, 0.0) {$a_4$};
\node[agentNode] (a5)  at (7.0, 0.0) {$a_5$};
\node[agentNode] (a6)  at (8.0, 0.0) {$a_6$};
\node[itemNode,bNode] (c3) at (6.5, 1.0) {$c_3$};
\node[itemNode,bNode] (c4) at (7.5, 1.0) {$c_4$};
\draw[heavyEdge] (c3) -- (a4b);
\draw[lightEdge] (c3) -- (a5);
\draw[heavyEdge] (c4) -- (a5);
\draw[lightEdge] (c4) -- (a6);
}}

\begin{scope}[xshift={1.25cm}, heavyEdge/.style={}, lightEdge/.style={}]
\pic{fullGraph};
\end{scope}

\draw[myArrow] (3.5, -0.6) -- node[right]{decompose} (3.5, -1.8);

\begin{scope}[yshift={-3.5cm}, heavyEdge/.style={}, lightEdge/.style={}]
\pic{decompGraph};
\end{scope}

\draw[myArrow] (3.5, -4.2) -- node[right]{round} (3.5, -5.4);

\begin{scope}[yshift={-7.1cm}, heavyEdge/.style={thick}, lightEdge/.style={thin,dotted}]
\pic{decompGraph};
\end{scope}

\draw[myArrow] (3.5, -7.8) -- node[right]{recompose} (3.5, -9.0);

\begin{scope}[xshift={1.25cm}, yshift={-10.5cm}, heavyEdge/.style={thick}, lightEdge/.style={thin,dotted}]
\pic{fullGraph};
\end{scope}

\end{tikzpicture}

\caption[Decomposing a tree into 3 subtrees and rounding them independently.]{%
Decomposing a tree into 3 subtrees based on the partition
$(\textcolor{textRed}{\{c_1\}}, \textcolor{textGreen}{\{c_2\}}, \textcolor{textBlue}{\{c_3, c_4\}})$ of the chores,
and then rounding the subtrees independently. Circles represent agents and squares represent chores.
}
\label{fig:decomp}
\end{figure}

\subsection{Formalizing the Decomposition}
\label{sec:lcr:tree-split-defn}

We start by formally defining our analogue of the tree-splitting technique,
and showing that it can be used to decompose the low-cost rounding problem into simpler sub-problems.

Let $x \in [0, 1]^{n \times m}$ be a fractional allocation for the instance
$\Ical \defeq ([n], [m], (d_i)_{i=1}^n, (w_i)_{i=1}^n)$ such that $G(x)$ is a forest.
For any $S \subseteq [m]$, define $\delta_x(S)$ to be the agents adjacent to some chore in $S$, i.e.,
\[ \delta_x(S) \defeq \{i \in [n]: x_{i,c} > 0 \text{ for some } c \in S\}. \]
Define $\Ical_x(S) \defeq (\delta_x(S), S, (d_i)_{i \in \delta_x(S)}, (w'_i)_{i \in \delta_x(S)})$
to be the instance $\Ical$ restricted to chores $S$, where
$(w'_i)_{i \in \delta_x(S)}$ is an arbitrary sequence of positive numbers that sum to 1.

Let $(M_1, \ldots, M_T)$ be a partition of the chores $[m]$.
Define the fractional allocation $x^{(t)}$ as the restriction of $x$ to $\Ical_x(M_t)$, i.e.,
$x^{(t)} \in [0, 1]^{\delta_x(M_t) \times M_t}$, where $x^{(t)}_{i,j} \defeq x_{i,j}$
for all $i \in \delta_x(M_t)$ and $j \in M_t$.
For any rounding $A$ of $x$, define the allocation $A^{(t)}$ as $(A_i \cap M_t)_{i \in \delta_x(M_t)}$.
Then
\begin{align*}
\rcost_{\Ical}(A \mid x) &= \sum_{i=1}^n (d_i(A_i) - d_i(x_i))^+
\\ &= \sum_{i=1}^n \max\left(0, \sum_{t=1}^T \left(d_i(A_i \cap M_t) - \sum_{c \in M_t} d_i(c)x_{i,c}\right)\right)
\\ &\le \sum_{t=1}^T \sum_{i \in \delta_x(M_t)}
    \max\left(0, d_i(A_i \cap M_t) - \sum_{c \in M_t} d_i(c)x_{i,c}\right)
\\ &= \sum_{t=1}^T \rcost_{\Ical_x(M_t)}(A^{(t)} \mid x^{(t)}).
\end{align*}
Hence, to round $x$, we can partition the chores and round each set in the partition independently.

For any $S \subseteq [m]$, define $G(x) \cap S$ to be
the subgraph of $G(x)$ induced by vertices $S \cup \delta_x(S)$.
We will often pick the partition $(M_1, \ldots, M_T)$ such that
$G(x) \cap M_t$ has special structure for each $t \in [T]$,
which would help us round it with a low cost.

\subsection{Decomposition Algorithm}
\label{sec:lcr:tree-split-algo}

Using \cref{sec:lcr:tree-split-defn}, we can assume \wLoG{} that $G(x)$ is a tree,
because if $G(x)$ is a forest, then we can just round each tree independently.

We can also assume that every chore has degree at least 2.
This is because if some chore $c$ is adjacent to just one agent $i$,
then $x_{i,c} = 1$, so it is already integrally allocated,
and we can remove it from consideration.

Since $G(x)$ is a tree, it has $m+n-1$ edges. Since all chores have degree at least 2,
we get
\[ m+n-1 = \sum_{c=1}^m \deg_{G(x)}(c) \ge 2m. \]

So we have $m \le n-1$.

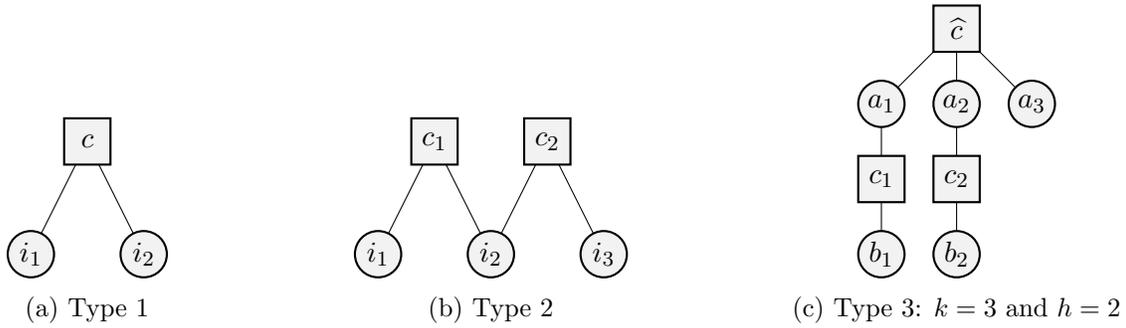
\begin{figure*}[htbp]
\centering
\begin{subfigure}[b]{0.25\textwidth}
\centering
\begin{tikzpicture}
\node[agentNode] (i1) at (0.00, 0.0) {$i_1$};
\node[itemNode]  (c)  at (0.75, 1.5) {$c$};
\node[agentNode] (i2) at (1.50, 0.0) {$i_2$};

\draw (i1) -- (c) -- (i2);
\end{tikzpicture}

\caption{Type 1}
\label{fig:type1}
\end{subfigure}
\hfill
\begin{subfigure}[b]{0.35\textwidth}
\centering
\begin{tikzpicture}
\node[agentNode] (i1) at (0.00, 0.0) {$i_1$};
\node[itemNode]  (c1) at (0.75, 1.5) {$c_1$};
\node[agentNode] (i2) at (1.50, 0.0) {$i_2$};
\node[itemNode]  (c2) at (2.25, 1.5) {$c_2$};
\node[agentNode] (i3) at (3.00, 0.0) {$i_3$};

\draw (i1) -- (c1) -- (i2) -- (c2) -- (i3);
\end{tikzpicture}

\caption{Type 2}
\label{fig:type2}
\end{subfigure}
\hfill
\begin{subfigure}[b]{0.35\textwidth}
\centering
\begin{tikzpicture}
\node[itemNode]  (chat) at (1, 3) {$\chat$};
\node[agentNode] (a1) at (0, 2) {$a_1$};
\node[agentNode] (a2) at (1, 2) {$a_2$};
\node[agentNode] (a3) at (2, 2) {$a_3$};
\node[itemNode]  (c1) at (0, 1) {$c_1$};
\node[itemNode]  (c2) at (1, 1) {$c_2$};
\node[agentNode] (b1) at (0, 0) {$b_1$};
\node[agentNode] (b2) at (1, 0) {$b_2$};

\draw (chat) -- (a1) -- (c1) -- (b1);
\draw (chat) -- (a2) -- (c2) -- (b2);
\draw (chat) -- (a3);
\end{tikzpicture}

\caption{Type 3: $k=3$ and $h=2$}
\label{fig:type3}
\end{subfigure}
\caption{Types of trees. Circles are agents and squares are chores.}
\label{fig:tree-types}
\end{figure*}

We show that a partition $(M_1, \ldots, M_T)$ of the chores always exists such that
for all $t \in [T]$, each subgraph $G(x) \cap M_t$ is one of the following types
(see \cref{fig:tree-types}):
\begin{enumerate}
\item Type 1: A single chore $c$ and two agents $i_1$ and $i_2$ who share it.
\item Type 2: Two chores $\{c_1, c_2\}$ and three agents $\{i_1, i_2, i_3\}$,
    where $i_1$ and $i_2$ share $c_1$, and $i_2$ and $i_3$ share $c_2$.
\item Type 3: A collection of chores $\{\chat, c_1, \ldots, c_h\}$ and a group of agents $\{a_1, \ldots, a_k, b_1, \ldots, b_h\}$,
    where $k \ge 3$ and $0 \le h \le k$. Chore $\chat$ is shared by agents $a_1$, \ldots, $a_k$,
    and chore $c_j$ is shared by agents $a_j$ and $b_j$ for $j \in [h]$.
\end{enumerate}

\begin{lemma}
\label{thm:split-deg-le-2}
Let $x \in [0, 1]^{n \times m}$ be a fractional allocation such that
$G(x)$ is a tree and every chore has degree 2.
Then we can compute a partition of the chores in $O(n^2)$ time such that
each induced subgraph is of type 1 or 2, and at most one subgraph is of type 1.
Moreover, a subgraph of type 1 exists iff $m$ is odd.
\end{lemma}
\begin{proof}
Let $G' = ([n], E')$ be a graph where $(i, j) \in E$ if agents $i$ and $j$ share a chore,
i.e., $\exists c \in [m]$ such that $x_{i,c} > 0$ and $x_{j,c} > 0$. Then $G'$ is a tree.
Using Lemma 4.4 of \citet{wu2024tree}, in $O(n^2)$ time, we get an edge-disjoint decomposition of $G'$
into subtrees with at most two edges, and at most one such subtree has exactly one edge.
Since edges correspond to chores, and all subtrees having type 2 implies $m$ is even,
the lemma stands proven.
\end{proof}

\begin{lemma}
\label{thm:split-deg-ge-3}
Let $x \in [0, 1]^{n \times m}$ be a fractional allocation such that
$G(x)$ is a tree and some chore has degree at least 3.
Then we can compute a partition of the chores in $O(m^2)$ time such that
each induced subgraph is of type 2 or 3.
\end{lemma}
\begin{proof}
We prove this using induction over the number of chores.
This is trivially true if there are 0 chores.
Now assume at least $m \ge 1$ chores exist,
and every subtree induced by at most $m-1$ chores
can be decomposed into subtrees of types 2 and 3
if some chore has degree at least 3.

Let $\chat$ be a chore of degree $k \ge 3$.
We will color some chores red such that red chores induce a subtree of type 3,
and we will partition the remaining chores such that the induced subtrees have types 2 and 3.
First, root the tree $G(x)$ at $\chat$, and color $\chat$ red.
Let $a_1, \ldots, a_k$ be the children of $\chat$.
Fix $i \in [k]$.
If a descendant chore of $a_i$ has degree at least 3,
decompose the subtree rooted at $a_i$ into
subtrees of types 2 and 3 using the induction hypothesis.

Now assume all of $a_i$'s descendant chores have degree 2.
Let chores $c_1, \ldots, c_{\ell}$ be $a_i$'s children,
and $T_j$ be the subtree rooted at $c_j$. Let
\[ J_{\odd} \defeq \{j \in [\ell]: T_j \text{ has an odd number of chores}\}. \]
For each $j \in [\ell] \setminus J_{\odd}$, decompose $T_j$ into
subtrees of type 2 using \cref{thm:split-deg-le-2}.
Pair up $2\floor{|J_{\odd}|/2}$ chores from $J_{\odd}$.
For each pair $(j_1, j_2)$, $T_{j_1} \cup T_{j_2} \cup \{a_i\}$
is a subtree of $G(x)$ and contains an even number of chores,
so we decompose it into subtrees of type 2 using \cref{thm:split-deg-le-2}.
Now at most one child $c^*$ of $a_i$ remains that hasn't been part of a decomposition.
Color $c^*$ red. Since $c^*$ has degree 2, it has exactly one child agent, which we denote as $b_i$.
The subtree rooted at $b_i$ has an even number of chores (perhaps zero chores),
so decompose it into subtrees of type 2 using \cref{thm:split-deg-le-2}.

Do this for all $i \in [k]$. Now all red chores induce a subtree of $G(x)$ of type 3.
All the remaining chores have been partitioned such that each induced subtree has type 2 or 3.
This completes the inductive step of the proof.
By mathematical induction, the lemma holds for any number of chores.
One can convert this inductive proof into a recursive algorithm that runs in $O(m^2)$ time.
\end{proof}

\subsection{Rounding Small Trees}
\label{sec:lcr:rounding-small-trees}

Let $\Ical$ be a bounded fair division instance having $n$ agents and $m$ chores,
where all agents have the same disutility function $d$.
Given a fractional allocation $x$ such that $G(x)$ is one of the three types
mentioned in \cref{sec:lcr:tree-split-algo}, we will show how to round it.

\begin{lemma}
\label{thm:type-1-round}
If $G(x)$ is a type 1 tree, we can round it with cost at most $1/2$.
\end{lemma}
\begin{proof}
Since $G(x)$ is a type 1 tree, a single chore $c$ is shared among agents 1 and 2.
\WLoG{}, assume $x_{1,c} \ge x_{2,c}$.
Since $x_{1,c} + x_{2,c} = 1$, we get $x_{1,c} \ge 1/2 \ge x_{2,c}$.
Allocate $c$ to agent 1. Then the rounding cost is
$d(c)(1 - x_{1,c}) \le d(c)/2 \le 1/2$.
\end{proof}

\begin{lemma}[Theorem 4.1 by \citet{wu2024tree}]
\label{thm:type-2-round}
If $G(x)$ is a type 2 tree, we can round it with cost at most $2/3$.
\end{lemma}

\begin{lemma}
\label{thm:min-max}
For any real numbers $a$, $b$, and $c$,
we have
\begin{equation*}
    r \defeq \min(c+a, \max(c, b)) \le \max(c, (a+b+c)/2).
\end{equation*}
\end{lemma}
\begin{proof}
If $b \le c$, then $r = \min(c+a, c) \le c$.
If $b \ge c$, then $r = \min(c+a, b) \le (a+b+c)/2$.
\end{proof}

\begin{lemma}
\label{thm:type-3-round}
If $G(x)$ is a type 3 tree, we can round it with cost at most $(k+h-1)/3$.
\end{lemma}
\begin{proof}
Recall that in a type 3 tree, the set of chores is $\{\chat, c_1, \ldots, c_h\}$,
and the set of agents is $\{a_1, \ldots, a_k, b_1, \ldots, b_h\}$.
For $i \le k$, define $y_i \defeq x_{a_i,\chat}$.
For $i \le h$, define $z_i \defeq x_{a_i,c_i}$.
We first characterize the subsidy required to round each $c_i$, for $i\leq h$, depending on whether agent $a_i$ receives $\chat$.

Suppose $i \le h$ and $a_i$ doesn't receive $\chat$.
If $a_i$ receives $c_i$, the cost of rounding $c_i$ is
\begin{equation*}
    \left( d(c_i)(1-z_i) - y_id(\chat) \right)^+.
\end{equation*}
If $a_i$ doesn't receive $c_i$, the cost of rounding $c_i$ is $d(c_i)z_i$.
By \cref{thm:min-max} (wth $c=0$), the smaller of these two costs is at most
\begin{equation}
\label{eqn:type-3-round:no-chat}
\max\left(0, \frac{d(c_i) - d(\chat)y_i}{2}\right) \le \frac{1 - d(\chat)y_i}{2}.
\end{equation}

Suppose $i \le h$ and $a_i$ receives $\chat$.
If $a_i$ also receives $c_i$, the cost of rounding $c_i$ is
\begin{equation*}
    (1-y_i)d(\chat) + (1-z_i)d(c_i).
\end{equation*}
If $a_i$ doesn't receive $c_i$, the cost of rounding $c_i$ is
\begin{equation*}
    ((1-y_i)d(\chat) - z_id(c_i))^+ + z_id(c_i) = \max((1-y_i)d(\chat), z_id(c_i)).
\end{equation*}
By \cref{thm:min-max}, the smaller of these two costs is at most
\begin{align}
&\max\left((1-y_i)d(\chat), \frac{(1-y_i)d(\chat)+d(c_i)}{2}\right)
\ifTwoColumn{\nonumber\\ &}{}
\le \frac{1 + (1-y_i)d(\chat)}{2}.
\label{eqn:type-3-round:chat}
\end{align}

Assume \wLoG{} that the sequence $(y_i)_{i=1}^h$ is non-increasing,
and the sequence $(y_i)_{i=h+1}^k$ is non-decreasing.
Let $p \defeq \sum_{i=1}^h y_i$ and $q \defeq \sum_{i=h+1}^k y_i$. Then $p + q = 1$.
If $h \ge 1$ and we allocate $\chat$ to agent $a_1$, the total cost of rounding is at most
\begin{align}
&\underbrace{\frac{1+(1-y_1)d(\chat)}{2}}_{\text{using \eqref{eqn:type-3-round:chat}}}
    + \sum_{i=2}^h \underbrace{\frac{1-y_id(\chat)}{2}}_{\text{using \eqref{eqn:type-3-round:no-chat}}}
\ifTwoColumn{\nonumber\\ &}{}
= \frac{h}{2} + d(\chat)\cdot \frac{1-p}{2}.
\label{eqn:type-3-round:1}
\end{align}

If $h < k$ and we allocate $\chat$ to agent $a_k$, the total cost of rounding is at most
\begin{align}
&\sum_{i=1}^h \underbrace{\frac{1-y_id(\chat)}{2}}_{\text{using \eqref{eqn:type-3-round:no-chat}}} + d(\chat)(1-y_k)
\ifTwoColumn{\nonumber\\ &}{}
= \frac{h}{2} + d(\chat)\left(-\frac{p}{2} + 1 - y_k\right).
\label{eqn:type-3-round:k}
\end{align}

We now get four cases depending on the value of $h$:
\begin{enumerate}
\item $h = 0$:
    Then $p = 0$ and $y_k \ge 1/k$. Allocate $\chat$ to agent $a_k$. Rounding cost is at most
    \begin{equation*}
        h/2 + d(\chat)(1-y_k) \le 1-1/k.
    \end{equation*}

\item $1 \le h \le k-2$:
    Allocate $\chat$ to agent $a_1$. Rounding cost is at most $(h+1)/2$.
\item $h = k-1$:
    Then $p = 1-y_k$. If we allocate $\chat$ to agent $a_1$, rounding cost is at most $h/2 + d(\chat)y_k/2$.
    If we allocate $\chat$ to agent $a_k$, rounding cost is at most $h/2 + d(\chat)(1-y_k)/2$.
    The better choice has rounding cost at most $h/2 + 1/4$.
\item $h = k$:
    Then $p = 1$. Allocate $\chat$ to agent $a_1$. Rounding cost is at most $h/2$.
\end{enumerate}
Hence, the rounding cost is at most
\begin{align*}
& \frac{h}{2} + \begin{cases}
    (k-1)/k & \text{ if } h = 0
    \\ 1/2 & \text{ if } 1 \le h \le k-2
    \\ 1/4 & \text{ if } h = k-1
    \\ 0 & \text{ if } h = k
    \end{cases}
\ifTwoColumn{\\ &}{\;}
\le \frac{h}{2} + \frac{k-1-h/2}{3} = \frac{k+h-1}{3}.
\qedhere
\end{align*}
\end{proof}

\subsection{Summing the Cost Over Subtrees}
\label{sec:lcr:sum}

Now that we have partitioned the chores and bounded the rounding cost for each induced subtree,
let us find the rounding cost for the entire instance.

\begin{lemma}
\label{thm:rcost-deg-le-2}
Let $x \in [0, 1]^{n \times m}$ be a fractional allocation for
a bounded fair division instance $\Ical$
where all agents have the same disutility function.
If $G(x)$ is a tree and every chore has degree 2,
we can compute a rounding $A$ of $x$ in $O(n^2)$ time such that
$\rcost_{\Ical}(A \mid x) \le n/3 - 1/6$.
\end{lemma}
\begin{proof}
Using \cref{thm:split-deg-le-2}, we can partition the chores in $O(n^2)$ time such that
each induced subgraph has type 1 or 2, and at most one subgraph has type 1.
By \cref{thm:type-1-round}, we pay a cost of $1/2$ for each chore in a type 1 subtree,
and by \cref{thm:type-2-round}, we pay a cost of $1/3$ for each chore in a type 2 subtree.
Hence, the total rounding cost is at most $m/3 + 1/6$.

Since each chore has degree 2, the number of edges is $2m$.
Since $G(x)$ is a tree, the number of edges is $m+n-1$.
Hence, $m = n-1$, so the total rounding cost is at most $n/3 - 1/6$.
\end{proof}

\begin{lemma}
\label{thm:rcost-deg-ge-3}
Let $x \in [0, 1]^{n \times m}$ be a fractional allocation for
a bounded fair division instance $\Ical$
where all agents have the same disutility function.
If $G(x)$ is a tree and some chore has degree at least 3,
we can compute a rounding $A$ of $x$ in $O(n^2)$ time such that
$\rcost_{\Ical}(A \mid x) \le (n-1)/3$.
\end{lemma}
\begin{proof}
Using \cref{thm:split-deg-le-2}, we can partition the chores in $O(n^2)$ time
such that each induced subgraph has type 2 or 3.
Define the weight of a tree to be the number of edges minus the number of chores.
By \cref{thm:type-2-round,thm:type-3-round}, the rounding cost of
each type-2 or type-3 subtree is at most $1/3$ of the subtree's weight.
Let $e$ be the number of edges in $G(x)$.
Then the total rounding cost is at most $(e - m)/3$.
Since $G(x)$ is a tree, $e = m + n - 1$,
so the total rounding cost is at most $(n-1)/3$.
\end{proof}

If $G(x)$ is a forest, we can round each tree separately,
so using \cref{thm:rcost-deg-le-2,thm:rcost-deg-ge-3},
we get that the total rounding cost is at most $n/3 - 1/6$.
This completes the proof of \cref{thm:lcr}.

\section{Conclusion}
\label{sec:conclusion}

We give a polynomial-time algorithm for allocating chores among
agents with additive disutilities and potentially unequal weights.
Our algorithm's output is not only fPO, but also matches the
best-known PROP-subsidy upper-bound of $n/3 - 1/6$ by \citet{wu2024tree}.

Our result is much simpler than \citet{wu2024tree}
due to our clever use of market equilibrium.
We used the payment vector to reduce the low-cost rounding problem to
a special case where agents have identical disutility functions.
This significantly simplified the rounding algorithm,
and we believe that this simplification can aid in the design of
algorithms with even better upper bounds on subsidy.

The PROP-subsidy problem has been studied for goods too.
The lower bound of $n/4$ continues to hold for goods,
and \citet{wu2024tree} give an algorithm whose subsidy is at most $n/3 - 1/6$.
We tried to extend our approach to goods too, but our reduction to
identical valuation functions only works for chores, not goods.
For chores, making disutilities proportional to payments involves \emph{decreasing} disutilities,
but for goods, we must \emph{increase utilities}, which can destroy the instance's boundedness.

More general versions of the problem offer interesting directions for future work,
like when agents's disutilities are not additive,
or when the items are a mix of goods and chores.
Another direction to explore is to consider fairness notions other than PROP,
like MMS \cite{budish2011combinatorial,procaccia2014fair} or APS \cite{babaioff2023fair}.
It is not known whether a sublinear subsidy suffices for these notions.

\end{document}